\documentclass[12pt]{article}

\usepackage[latin1]{inputenc}
\usepackage[T1]{fontenc}
\usepackage[english]{babel}
\usepackage{amsfonts}
\usepackage{amssymb}
\usepackage{amsmath}
\usepackage{amsthm}
\usepackage{graphicx,xcolor}
\usepackage{vmargin}

\newtheorem{lemma}{Lemma}
\newtheorem{teorema}{Theorem}
\newtheorem{corollario}{Corollary}

\setpapersize{A4}
\setmarginsrb{15mm}{10mm}{15mm}{10mm}
             {0mm}{10mm}{0mm}{10mm}

\def\be{\begin{equation}}
\def\ee{\end{equation}}
\def\bea{\begin{eqnarray}}
\def\eea{\end{eqnarray}}
\def\ni{\noindent}
\def\nn{\nonumber}

\def\dist{\operatorname{dist}}

\def\z{\zeta}
\def\f{\varphi}
\def\b{\beta}
\def\g{\gamma}
\def\e{\eta}
\def\d{\delta}
\def\L{\Lambda}
\def\r{\rho}

\def\bq{\bar q}
\def\RRR{\mathbb{R}}

\def\MM{{\bf M}}

\def\RRR{{\mathbb R}}

\def\EE{{\cal E}}

\def\AA{{\cal A}}
\def\HH{{\cal H}}
\def\MM{{\cal M}}

\def\0{\emptyset}
\def\L{\Lambda}

\def\n{\nu}
\def\b{\beta}
\def\m{\mu}
\def\d{\delta}

\def\g{\gamma}
\def\r{\rho}

\def\e{\varepsilon}

\def\G{\Gamma}

\def\f{\varphi}
\def\o{\omega}
\def\c{\chi}
\def\O{\Omega}

\def\ol{\overline}

\title{On the stationary BBGKY hierarchy for equilibrium states}

\author{Giuseppe Genovese, Sergio Simonella\\  \\ \textit{Dipartimento di Matematica, Sapienza Universit\`a di Roma}\\
\textit{Piazzale Aldo Moro, 2, 00185 Roma, Italia}}
\date{}   

\begin{document}

\maketitle

\begin{abstract}
We consider infinite classical systems of particles interacting via a smooth, stable and regular two--body potential. 
We establish a new direct integration method to construct the solutions of the stationary BBGKY hierarchy,
assuming the usual Gaussian distribution of momenta. We prove equivalence between the corresponding infinite 
hierarchy and the Kirkwood--Salsburg equations.
A problem of existence and uniqueness of the solutions of the hierarchy
with appropriate boundary conditions is thus solved for low densities. 
The result is extended in a milder sense to systems with a hard core interaction.
\end{abstract}

\section{Introduction}

The BBGKY hierarchy is the fundamental system of equations for the evolution of correlation functions
of a state in Statistical Mechanics, \cite{Bo}. It provides the first bridge between the newtonian description 
of motion and the statistical description of macroscopic systems. For a very large system of particles 
obeying the Newton laws, the set of infinite coupled integro--differential equations is generally
used for the time evolution at least in case of sufficiently smooth correlations (\cite{Cohen}). 
The 70--years--old history of the hierarchy has brought enormous
progress in the investigation of the transition from the microscopic to the macroscopic world:
great advance has come by suitable methods of truncation, approximation and scaling limits, giving
a justification of the kinetic equations that describe particle systems on the mesoscopic (intermediate) scales.

We know that the non truncated BBGKY may contain informations of great relevance that are 
eventually lost in the mesoscopic limits. On the other hand, the complex mathematical structure 
of the hierarchy makes quite hard to use the system of equations in its entirety. For this, despite 
the importance and the long standing history of the problem, it is still necessary to develop new techniques.
In this paper we focus on the problem of the solution of the complete hierarchy in simple cases.
Namely the integration of the equations for an infinite classical system of particles at equilibrium.
This means that we investigate the form that the hierarchy assumes when one seeks stationary 
solutions with a Maxwellian distribution of momenta: an infinite system relating the positional 
correlation functions.

A first attempt in this direction is in the pioneering and remarkable paper by Morrey \cite{Mo55}. 
In that work, a long and yet involved proof leads directly from the BBGKY to some non trivial series expansion, 
which is then shown to be convergent for small densities. Later on, Gallavotti and Verboven \cite{GV75}
try to give a more transparent proof of the theorem of Morrey, dealing with the case of smooth, bounded and 
short range potentials. By direct integration and iteration they derive the Kirkwood--Salsburg 
equations, a set of integral relations which is well known to be one of the several characterizations
of an equilibrium state. To do so, they need assumptions of small density, exponential strong cluster properties, 
and rotation and translation symmetry of the state. Finally, in a series of four papers \cite{GS76}, Gurevich and 
Suhov, with different methods, establish that any Gibbs states (not necessarily Maxwellian) with potential 
belonging to a certain general class and satisfying the stationary BBGKY, is an equilibrium DLR state
associated with the interaction appearing in the hierarchy. This is done for infinite systems of particles over all space with a 
smoothed hard core and short range interaction. Their analysis is focused on the dual hierarchy satisfied by
the generating functions of the Gibbs state and on the study of first integrals of the corresponding hamiltonian.

Our main aim in this paper is to provide, following the program of \cite{Mo55} and \cite{GV75}, a simple, short and 
self--contained method of direct integration of the infinite system of BBGKY equations 
(assuming Maxwellian distribution of momenta) with boundary conditions, 
suitable to be applied to situations more general than those treated in \cite{GV75}. In particular, we want
to extend the results to any stable, long range, possibly singular potential, with some weak decrease property 
assuring the usual regularity required for thermodynamic stability. Moreover, we shall work out a direct iterative procedure 
that does not require small density and rotation invariance assumptions. Also, we decouple the problem of integration
from that of imposing boundary conditions. This last feature allows to weaken the hypotheses of clustering and is the key 
point to extend the results to different kinds of boundary condition. As an example of this, we apply the method to the 
case of non symmetric states, for which translation invariance is assumed only at infinity, such as equilibrium
states of particles in containers with walls extending to infinity. In the cases treated here, together with the 
equivalence with the KS equations, the method gives uniqueness of the solution in the small density--high 
temperature region (resorted to the uniqueness of the solution of the Kirkwood--Salsburg equations).
We hope that the methods developed may help in understanding the structure of nonequilibrium
stationary states, along the line of research originally introduced in \cite{L59}, in which substantial progress has 
not been reached yet.

The only relevant case left open in the Maxwellian framework of this paper,
is that of the hard core interactions, for which the proof presented here does not directly apply. 
The additional difficulty is due to the presence of ``holes'' in the phase space. 
However, we will point out that the procedure established in the proof is uniform in the hard core approximation.
Thus, the classical equilibrium solution of the hard core BBGKY (defined by the corresponding KS equations)
can be uniquely determined as a limit of solutions of smooth hierarchies in a fixed space of states,
with few restrictions on the form of the approximating potentials.

The paper is organized as follows: in Section \ref{sec:setup} we introduce the infinite system of particles
and the infinite system of equations we will deal with; in Section
\ref{sec:mr} we state our results on the integration of the stationary hierarchy, and discuss the proof;
in Section \ref{sec:hc} we analyse the solutions of the hard core equilibrium hierarchy.
Some additional notes are deferred to the Appendices. In particular, in Appendix B the integration 
problem is solved for the one--dimensional hard core system. Finally, we dedicate Appendix C 
to revise the argument of \cite{GV75}: we point out an error in the procedure, explain 
how to correct it and make comparisons with the new method.


\section{Setup} \label{sec:setup}

We will consider an infinite classical system of particles with unitary mass, interacting through 
a smooth stable pair potential, possibly diverging at the origin and with some weak decrease property at infinity. 
We will introduce a class of measures over the phase space of the system, with features
that assure the existence of correlations. Then we will write the infinite hierarchy of 
equations satisfied by them, assuming a Maxwellian distribution for the velocities, as well as smoothness
of the correlations. We list below the definitions required:

1) The {\em phase space} $\HH$
is given by the infinite countable sets 
$X = \{x_i\}_{i=1}^{\infty} \equiv \{(q_i,p_i)\}_{i=1}^{\infty}$, $x_i\in\RRR^{\nu}\times\RRR^{\nu},
\nu =1,2,3,$ which are locally finite: $\L\cap \left(\cup_{i=1}^{\infty}q_i\right)$ 
is finite for any bounded region $\L \subset \RRR^\nu$.

2) The {\em Hamiltonian} of the system is defined by the formal function on $\HH$
\be
H(X) = \sum_{i=1}^{\infty} \frac{p_i^2}{2} + \sum_{i<j}^{1,\infty}\f(q_i-q_j)\;, 
\label{eq:Ham}
\ee
where $\f:\RRR^\nu\setminus\{0\} \rightarrow \RRR$ is assumed to be radial 
and such that, for some $B>0$ and every $c>0,$
\bea
&& \sum_{i<j}^{1,n}\f(q_i-q_j) \geq -nB\ \ \ \ \ \forall\ n\geq 0,\ \ \ \ \ q_1,\cdots,q_n\in\RRR^\n
\ \ \ \ \ \mbox{(stability)}\;; \label{eq:stab}\\
&& \int_{\RRR^\n} dq|1-e^{-c\f(q)}| +  \int_{\RRR^\n} dq|\nabla(1-e^{-c\f(q)})| < +\infty
\ \ \ \ \ \ \ \ \ \mbox{(decrease)}\;;\label{eq:decrease} \\
&& e^{-c\f} \in C^1(\RRR^\n)\ \ \ \ \ \ \ \ \ \ \ \ \ \ \ \ \ \ \ \ \ \ \ \ \ \ \ \ \ \ \ \ \ \ \ \ \ \ \ \ \ \ \ \ \ \ \ \ \ \ \ \ 
\mbox{(smoothness)}\;.\label{eq:smoothness}
\eea
Condition (\ref{eq:stab}) implies that $\f$ is bounded from below by $-2B,$ while 
the stated decrease property is equivalent to absolute integrability of $\f$ and $|\nabla\f|$ outside any ball centered in the origin.
The smoothness condition ensures that $\f$ is $C^1$ outside the origin, and that the singularity at the origin (if any) is approached
not too slowly (for instance logarithmic divergences are excluded); see the Remark on page \pageref{thm:main} for a discussion
on the regularity conditions.

3) A {\em state} is a probability measure $\m$ on the Borel sets of $\HH:$ see \cite{Ru69}, \cite{Ru67}. 
Following \cite{Ru67}, we may define it as a collection $\{\m_{\L}\}$ of probability measures 
on $\HH_{\L}:= \oplus_{n=0}^{\infty} (\L\times\RRR^\nu)^n,$ $\L\subset\RRR^\nu$ bounded open,
satisfying the following properties:
\begin{description}
\item[a.] \ \ \ the restriction $\m_{\L}$ to the space $(\L\times\RRR^\nu)^n$ is absolutely continuous
with respect to Lebesgue measure, with a density of the form $\frac{1}{n!}\m_{\L}^{(n)}(x_1,\cdots,x_n)$, symmetric 
for exchange of particles;
\item[b.] \ \ \ $\m_{\emptyset}^0(\RRR^0\times\RRR^0)=1$;
\item[c.] \ \ \ if $\L \subset \L'$, then
\be
\m_{\L}^{(n)}(x_1,\cdots,x_n) = \sum_{p=0}^{\infty}\frac{1}{p!} \int_{((\L'\setminus\L)\times\RRR^\nu)^p}
dx_{n+1}\cdots dx_{n+p} \m_{\L'}^{(n+p)}(x_1,\cdots,x_{n+p})\;;
\ee
\item[d.] \ \ \ $\m_{\L}^{(n)}\leq C^n_{\L}\prod_{i=1}^n \eta_{\L}(|p_i|)$ for some constant $C_{\L}$
and $\eta_{\L}(|p|)\in L^1(\RRR^\nu)$, so that the expression in the right hand side of the following equation
is well defined:
\be
\ol\r_n(x_1,\cdots,x_n) := \sum_{p=0}^{\infty}\frac{1}{p!}\int_{(\L\times\RRR^\nu)^p} dx_{n+1}\cdots dx_{n+p} 
\m_{\L}^{(n+p)}(x_1,\cdots,x_{n+p})\;, \label{eq:defcf}
\ee
where we assume $q_1,\cdots,q_n\in\L;$ 
\item[e.] \ \ \ there exist $\ol\xi>0, \eta(|p|)\in L^1(\RRR^\nu),$
such that  
\be
\ol\r_n(x_1,\cdots,x_n)\leq \ol\xi^n\prod_{i=1}^n \eta (|p_i|)\;. \label{eq:olrnbound}
\ee
\end{description}
Equation (\ref{eq:defcf}) defines the {\em correlation functions} of the state.

The state is said to be {\em invariant} if
\be
\m^{(n)}_\L(x_1,\cdots,x_n) = \m^{(n)}_{\L+a}(q_1+a,p_1,\cdots,q_n+a,p_n)  \label{eq:definv}
\ee
for all $a\in\RRR^\nu$ and $\L.$ 

{\em Remarks.} (i) Condition {\bf b.}, together with the compatibility condition {\bf c.}, imply the normalization
of the measures $\m_\L,$ i.e. 
\be
\sum_{n\geq 0}\frac{1}{n!}\int_{(\L\times\RRR^\nu)^n} dx_1\cdots dx_n
\m_{\L}^{(n)}(x_1,\cdots,x_n) = 1\;. \label{eq:norm}
\ee
\noindent (ii) Condition {\bf e.} guarantees convergence of the inverse formula
\be
\m_{\L}^{(n)}(x_1,\cdots,x_n) = \sum_{p=0}^{\infty}\frac{(-1)^p}{p!}\int_{(\L\times\RRR^\nu)^p}
dx_{n+1}\cdots dx_{n+p} \ol\r_{n+p}(x_1,\cdots,x_{n+p})\;;
\ee
hence the definition of correlation functions of the state is well posed. (iii) The definition (\ref{eq:defcf}) implies
that an invariant state has also translation invariant correlation functions.

Finally, we say that a state is {\em smooth Maxwellian} (with parameter $\b$ and Hamiltonian $H$)
when there exist $g_n:\RRR^{\nu n} \rightarrow\RRR^+, \b>0$ and $\xi>0, C_n>0$
such that the correlation functions have the form 
\bea
&&\ol\r_n(x_1,\cdots,x_n) = \prod_{i=1}^n\left(\frac{e^{-\b p_i^2/2}}{(2\pi/\b)^\nu}\right)\r_n(q_1,\cdots,q_n)\;, \nn\\
&&\r_n(q_1,\cdots,q_n) = e^{- \b \sum_{i<j}^{1,n}\f(q_i-q_j)} g_n(q_1,\cdots,q_n)\;,
\ \ \ \ \ \ \ \ \ g_n\in C^1(\RRR^{\n n})\;, \label{eq:Maxwell} 
\eea
with the following explicit bounds:
\bea
&& \r_n(q_1,\cdots,q_n) \leq \xi^n e^{-\b W_{q_i}(q_1,\cdots,q_{i-1},q_{i+1},\cdots,q_n)}\;,\label{eq:rnbound}\\
&& |\nabla_{q_i}\r_n(q_1,\cdots,q_n)| \leq C_n \xi^n e^{-\b W_{q_i}(q_1,\cdots,q_{i-1},q_{i+1},\cdots,q_n)}\;,
\ \ \ \ \ \ \ \ i = 1,\cdots,n\;, \label{eq:gradrnbound} 
\eea
where
\be
W_{q}(q_1,\cdots,q_m) = \sum_{i=1}^m \f(q-q_i)\;.
\ee
Equations (\ref{eq:rnbound}) and (\ref{eq:gradrnbound}) imply also 
$|\nabla_{q_i}\left(e^{+\b W_{q_i}(q_1,\cdots,q_{i-1},q_{i+1},\cdots,q_n)}\r_n(q_1,\cdots,q_n)\right)| \leq C'_n \xi^n$
for some $C'_n>0.$ Notice that this definition of smooth Maxwellian state is equivalent to the one without the exponentials
in formulas (\ref{eq:Maxwell})--(\ref{eq:gradrnbound}) in the case $\f$ is a $C^1$ function over all $\RRR^{\n}$
(i.e. a bounded function: hence in that case we recover the definition used in \cite{GV75}). 

4) A smooth Maxwellian state with parameter $\b$ is a {\em stationary solution} of the {\em BBGKY hierarchy} 
of equations with Hamiltonian $H$ if 
\bea
&&\nabla_{q_1} \r_n(q_1,\cdots,q_n)=-\b\Big[\nabla_{q_1}W_{q_1}(q_2,\cdots,q_n)\r_n(q_1,\cdots,q_n)\nn\\
&&\ \ \ \ \ \ \ \ \ \ \ \ \ \ \ \ \ \ \ \ \ + \int_{\RRR^\n}dy \nabla_{q_1} \f(q_1-y)\r_{n+1}(q_1,\cdots,q_n,y)\Big]\;,
\ \ \ \ \ \ \ \ \ \ n \geq 1\;,\label{eq:grad-rho}
\eea
for any $q_1,\cdots,q_n\in \RRR^{\nu}.$ 
For smooth Maxwellian states, this is equivalent to say that the $\ol\r_n$ solve the complete form 
of the stationary Bogolyubov equations
\bea
&& \sum_{i=1}^n \left[p_i\cdot \nabla_{q_i}\ol\r_n(x_1,\cdots,x_n) - \nabla_{q_i}W_{q_i}(q_1,\cdots q_{i-1},
q_{i+1},\cdots,q_n) \cdot \nabla_{p_i}\ol\r_n(x_1,\cdots,x_n)\right]\nn\\
&&\ \ \ \ \ =\sum_{i=1}^n \int_{\RRR^{\nu}\times\RRR^{\nu}}d\xi d\pi  \nabla_{q_i}\f(q_i-\xi)\cdot 
\nabla_{p_i}\ol\r_{n+1}(x_1,\cdots,x_n,\xi,\pi)\;,\ \ \ \ \ \ \ \ \ \ n \geq 1\;, \label{eq:BBGKYcomplete}
\eea
for all $x_1,\cdots,x_n\in \RRR^{2\nu},$ as it can be immediately verified using the arbitrariness of $p_1,\cdots,p_n.$


\section{Main results} \label{sec:mr}

Our first task is to solve Equation (\ref{eq:BBGKYcomplete}), provided the assumption that the state is 
translation invariant and smooth Maxwellian with parameter $\b>0,$ \cite{Mo55}.
Thus, as pointed out above, the problem is equivalent to consider the infinite system (\ref{eq:grad-rho})
and find a solution $\r_n\in C^1(\RRR^{\n n})$ symmetric in the exchange of particle labels, 
translation invariant and bounded as in (\ref{eq:rnbound})--(\ref{eq:gradrnbound}).
The equations are then parametrized by the two strictly positive constants $\r \equiv \r_1(q_1),$ and $\b>0$
(the Eq. (\ref{eq:grad-rho}) for $n=1$ will be useless in our assumptions). Once discussed this 
problem (Theorem \ref{thm:main} below), we will give a generalization of our result to the case in which 
the system of particles is contained in certain unbounded subsets of $\RRR^\n$ 
and translation invariance holds only at infinity (Theorem \ref{thm:infcont}).

To achieve the integration of Eq. (\ref{eq:grad-rho}) we need to add some boundary condition. We choose the
{\em cluster property} defined as follows. Denote $A_n$ and $B_m$ any two 
disjoint clusters of $n$ and $m$ points respectively in $\RRR^\n,$ 
such that $A_n\cup B_m = (q_1,\cdots,q_{n+m}).$ Indicate
$\dist(A_n,B_m) = \inf\{|q_i-q_j|; q_i\in A_n, q_j\in B_m\}.$ Then there exists a constant $C>0$
and a monotonous decreasing function $u$ vanishing at infinity such that
\bea
|\r_{n+m}(A_n, B_m) - \r_n(A_n)\r_m(B_m)| \leq C^{n+m} u(\dist(A_n,B_m))\;. \label{eq:DIS}
\eea
This is known to be satisfied by every equilibrium state for the considered class of potentials,
at least for sufficiently small density (small $\r$) and high temperature (small $\b$): see for instance \cite{Ru69}.

Our main result is the following
\begin{teorema} \label{thm:main}
If a smooth Maxwellian invariant state is a stationary solution of the BBGKY hierarchy with cluster
boundary conditions, then there exists a constant $z$ such that the correlation functions of the state
satisfy
\bea
&&\r_n(q_1,\cdots,q_n)= z e^{-\b W_{q_1}(q_2,\cdots,q_n)}\Big[\r_{n-1} (q_2,\cdots,q_n)\label{eq:KS}\\
&&\ \ \ \ \ \ \ \ \ \ +\sum_{m=1}^{\infty}\frac{(-1)^m}{m!}\int_{\RRR^{m\n}} dy_1 \cdots dy_m 
\prod_{j=1}^m \left(1-e^{-\b\f(q_1-y_j)}\right) \r_{n-1+m} (q_2,\cdots,q_n,y_1,\cdots,y_m)\Big]\;.\nn
\eea
Conversely, a smooth Maxwellian state with parameter $\b$ and satisfying (\ref{eq:KS}) is a stationary 
solution of the BBGKY hierarchy.
\end{teorema}

The integral relations (\ref{eq:KS}) are called the {\em Kirkwood--Salsburg equations}.
The series in the right hand side is absolutely convergent uniformly in $q_1,\cdots,q_n$
since $\f$ satisfies (\ref{eq:decrease}) and $\r_n\leq(\xi e^{2\b B})^n.$
We shall point out that, for $n=1,$ the first term in the right hand side has to be interpreted as $z;$ 
the Equation is in this case independent of $q_1$ by translation invariance: it
provides a definition of $z$ in terms of integrals of all the correlation functions.
Formula (\ref{eq:KS}) is one of the several characterizations of an {\em equilibrium} state for small density and high temperature,
and $z$ is identified with the {\em activity} of the system, e.g. \cite{GaSM}.

{\em Remark (regularity).} The regularity assumptions on the state and the potential made in Section~\ref{sec:setup} can 
be somewhat relaxed. For instance, we could require that $g_n=e^{+ \b \sum_{i<j}^{1,n}\f(q_i-q_j)} \r_n$ 
(or even just $e^{+\b W_{q_i}(q_1,\cdots,q_{i-1},q_{i+1},\cdots,q_n)} \r_n, i=1,\cdots,n$) has the same 
regularity as the function $e^{-\b\f},$ and piecewise continuity and boundedness of the derivative of $\f$ 
outside the origin (instead of $C^1$ regularity). What is strictly necessary in order to work out the proof below,
is a combined condition on $\f$ and $\r;$ namely, that Eq. (\ref{eq:hatBBGKY}) below holds for \textit{all} values of the 
first variable (so that the integration over straight lines can be performed, as in formula (\ref{eq:inthat})).
This feature prevents us to apply the method to the case of potentials having a hard core (both pure hard core and 
smoothed versions of it), unless we specify the correct boundary conditions for the functions $g_n.$ But this seems to 
be not trivial. For instance, one might think to set $g_n=0$ (or $g_n$ equal to a fixed constant value) 
inside the cores and on the boundaries. In this case, the proof of our theorem can be applied to show that, if there were 
such a solution, it would have to satisfy the Kirkwood--Salsburg equations, which in turn have a unique solution for small 
density and high temperature. But the latter does not satisfy the above prescription on $g_n,$
as it can be easily checked by direct calculation of the coefficients of the Mayer expansion.
Hence a solution with such a prescription cannot exist.
The problem of finding minimal boundary conditions for the $g_n$ that ensure existence and uniqueness 
(and equivalence with the Kirkwood--Salsburg equations) in presence of hard cores is open.

\bigskip

\begin{proof}[Proof of Theorem \ref{thm:main}] We prove here the direct statement.
The proof of the converse statement (which has been given also in \cite{Ga68}) 
is analogous to the one of Lemma \ref{lem:converse}, which will be discussed in Section \ref{sec:hc}.

First, we rewrite (\ref{eq:grad-rho}) as
\bea
&& e^{\b W_{q_1}(q_2,\cdots,q_n)}\Big(\nabla_{q_1} \r_n(q_1,\cdots,q_n)+\b \r_n(q_1,\cdots,q_n)
\nabla_{q_1}W_{q_1}(q_2,\cdots,q_n)\Big)\nn\\
&& = -\b \int_{\mathbb{R}^3}dy \nabla_{q_1} \f(q_1-y) e^{\b W_{q_1}(q_2,\cdots,q_n)}
\r_{n+1}(q_1,\cdots,q_n,y)\;,
\eea
so that the left hand side is equal to $\nabla_{q_1}\left( \r_ne^{\b W_{q_1}}\right).$
For the sake of clearness, hereafter we will put 
\bea
&&\hat\r_n(q_1;q_2,\dots,q_n):= e^{\b W_{q_1}(q_2,\dots,q_n)}\r_n(q_1,\dots,q_n)\;;\nn\\
&&K_{q_0q_1}(q,y):= (1-e^{-\b\f(q-y)})-(1-e^{-\b\f(q_1-y)})-(1-e^{-\b\f(q_0-y)})\;.\label{eq:defK}
\eea
As functions of the variable $q_1,$ the $\hat\r_n$ are of class $C^1(\RRR^\n)$ for any choice of $q_2,\cdots,q_n\in\RRR^\n.$
They satisfy
\bea
\nabla_{q_1}\hat\r_n(q_1;q_2,\dots,q_n)
= -\int_{\RRR^\n} dy_1 \nabla_{q_1}
\left(1-e^{-\b\f(q_1-y_1)}\right)\hat\r_{n+1}(q_1;q_2,\dots,q_n,y_1)\label{eq:hatBBGKY}
\eea
for any $q_1,\cdots,q_n\in \RRR^{\nu}.$ 

Fix $q_0\in\RRR^\n$ arbitrarily. We shall integrate the previous equation along a straight line 
$\overrightarrow{q_0q_1}$ connecting $q_0$ to $q_1.$ Using (\ref{eq:defK}) we deduce
\bea
\hat\r_n(q_1;q_2,\dots,q_n)-\hat\r_n(q_0;q_2,\dots,q_n) =-\int_{q_0}^{q_1} d\bar q_1\int_{\RRR^d} dy_1
\frac{\partial K_{q_0q_1}}{\partial \bar q_1}(\bar q_1,y_1)\hat\r_{n+1}(\bar q_1;q_2,\dots,q_n,y_1)\;,
\label{eq:inthat}
\eea
where $\int_{q_0}^{q_1} d\bar q_1$ and $\frac{\partial}{\partial \bar q_1}$ denote respectively
integration and differentiation along the straight line. Interchanging the integrations in the right hand side 
and integrating by parts we find
\bea
&&\hat\r_n(q_1;q_2,\dots,q_n)-\hat\r_n(q_0;q_2,\dots,q_n)\label{eq:N0}\\
&&= -\int dy_1\Big[ -\left(1-e^{-\b\f(q_0-y_1)}\right)\hat\r_{n+1}(q_1;q_2,\dots,q_n,y_1)\nn\\
&&\ \ \ \ \ \ \ \ \ \ +  \left(1-e^{-\b\f(q_1-y_1)}\right)\hat\r_{n+1}(q_0;q_2,\dots,q_n,y_1)  \Big]\nn\\
&&\ \ \ \ \ +\int_{\RRR^d} dy_1\int_{q_0}^{q_1} d\bar q_1
K_{q_0q_1}(\bar q_1,y_1)\frac{\partial \hat\r_{n+1}}{\partial \bar q_1}(\bar q_1;q_2,\dots,q_n,y_1)\;.\nn
\eea
All the above integrals are absolutely convergent thanks to (\ref{eq:decrease}), (\ref{eq:rnbound}) and (\ref{eq:gradrnbound}).

In the last term of the above equation we may iterate the projection of (\ref{eq:hatBBGKY}) 
along $\overrightarrow{q_0q_1},$ that can be written as
\bea
\frac{\partial\hat\r_n(\bar q_1;q_2,\dots,q_n)}{\partial \bar q_1}
= -\int_{\RRR^d} dy_1\frac{\partial K_{q_0q_1}}{\partial \bar q_1}(\bar q_1,y_1) \hat\r_{n+1}(\bar q_1;q_2,\dots,q_n,y_1)\;.
\label{eq:hatBBGKYp}
\eea
The last term of (\ref{eq:N0}) then becomes, proceeding as after (\ref{eq:inthat}),
\bea
&&-\int_{\RRR^d} dy_1\int_{\RRR^d} dy_2\int_{q_0}^{q_1} d\bar q_1
K_{q_0q_1}(\bar q_1,y_1)\frac{\partial K_{q_0q_1}}{\partial \bar q_1}(\bar q_1,y_2)
\hat\r_{n+2}(\bar q_1;q_2,\cdots,q_n,y_1,y_2)\;\nn\\
&&= -\frac{1}{2}\int_{\RRR^d} dy_1\int_{\RRR^d} dy_2
\Big[ \prod_{j=1,2}\left(1-e^{-\b\f(q_0-y_j)}\right)\hat\r_{n+2}(q_1;q_2,\dots,q_n,y_1,y_2)\nn\\
&&\ \ \ \ \ \ \ \ \ \ -  \prod_{j=1,2}\left(1-e^{-\b\f(q_1-y_j)}\right)\hat\r_{n+2}(q_0;q_2,\dots,q_n,y_1,y_2)\Big]\nn\\
&&\ \ \ \ \ +\frac{1}{2}\int_{\RRR^d} dy_1\int_{\RRR^d} dy_2\int_{q_0}^{q_1} d\bar q_1
\prod_{j=1,2}K_{q_0q_1}(\bar q_1,y_j)
\frac{\partial \hat\r_{n+2}}{\partial \bar q_1}(\bar q_1;q_2,\dots,q_n,y_1,y_2)\;,
\eea
having used also the symmetry for exchange of particles to perform the integration by parts.
We may iterate again (\ref{eq:hatBBGKYp}) in the last term of this formula.
After $N$ integrations by parts ($N$ iterations) we have
\bea
&&\hat\r_n(q_1;q_2,\dots,q_n)-\hat\r_n(q_0;q_2,\dots,q_n)\nn\\
&&=-\sum_{k=1}^{N}\frac{(-1)^k}{k!}
\int_{\RRR^{\n k}} dy_1\cdots dy_k \Big[\prod_{j=1}^k\Big(1-e^{-\b\f(q_0-y_j)}\Big)
\hat\r_{n+k}(q_1;q_2,\cdots,q_n,y_1,\cdots,y_k)\nn\\
&&\ \ \ \ \ \ \ \ \ \ -\prod_{j=1}^k\Big(1-e^{-\b\f(q_1-y_j)}\Big)
\hat\r_{n+k}(q_0;q_2,\cdots,q_n,y_1,\cdots,y_k)\Big]\nn\\
&&\ \ \ \ \ -\frac{1}{N!}\int_{\RRR^{\n(N+1)}} dy_1\cdots dy_{N+1}\int_{q_0}^{q_1} d\bar q_1\prod_{j=1}^N K_{q_0q_1}(\bar q_1,y_j)\nn\\
&&\ \ \ \ \ \ \ \ \ \ \cdot\frac{\partial K_{q_0q_1}}{\partial \bar q_1}(\bar q_1,y_{N+1})
\hat\r_{n+N+1}(\bar q_1;q_2,\cdots,q_n,y_1,\cdots,y_{N+1})\;\label{eq:prima di rem}.
\eea

Assumption (\ref{eq:rnbound}) allows to bound explicitly the last term with 
\bea
\frac{1}{N!}|q_1-q_0|\left(3\int_{\RRR^\n}dy\left(1-e^{-\b\f(y)}\right)\right)^N
\left(\int_{\RRR^\n}dy|\nabla(1-e^{-\b\f(y))}|\right) \xi^{n+1+N}\;.\label{eq:rem}
\eea
Thus by taking $N\rightarrow\infty$, it follows that,
for any arbitrary $q_0\in\RRR^\n,$ the correlation functions satisfy the set of integral equations
\bea
&&e^{-\b W_{q_0}(q_2,\cdots,q_n)}\sum_{k=0}^{\infty}\frac{(-1)^k}{k!}
\int_{\RRR^{\n k}} dy_1\cdots dy_k \prod_{j=1}^k\Big(1-e^{-\b\f(q_0-y_j)}\Big)\nn\\
&&\ \ \ \ \ \cdot e^{\b W_{q_1}(y_1,\cdots,y_k)}\r_{n+k}(q_1,q_2,\cdots,q_n,y_1,\cdots,y_k)\nn\\
&&= e^{-\b W_{q_1}(q_2,\cdots,q_n)}\sum_{k=0}^{\infty}\frac{(-1)^k}{k!}
\int_{\RRR^{\n k}} dy_1\cdots dy_k \prod_{j=1}^k\Big(1-e^{-\b\f(q_1-y_j)}\Big)\nn\\
&&\ \ \ \ \ \cdot e^{\b W_{q_0}(y_1,\cdots,y_k)}\r_{n+k}(q_0,q_2,\cdots,q_n,y_1,\cdots,y_k)\;, \label{eq:KSq0}
\eea
where the series in both sides are absolutely convergent uniformly in $q_0,q_1,\cdots,q_n.$

What is left in order to complete the proof is just taking the limit as $|q_0| \rightarrow \infty$ of (\ref{eq:KSq0}). 
Using the translation invariance of the correlation functions we have
\bea
&&\r_n(q_1,\cdots,q_n)  \sum_{k=0}^{\infty}\frac{(-1)^k}{k!}
\int_{\RRR^{\n k}} dy_1\cdots dy_k \prod_{j=1}^k\Big(1-e^{-\b\f(y_j)}\Big)\r_{k}(y_1,\cdots,y_k)\nn\\
&&=\r e^{-\b W_{q_1}(q_2,\cdots,q_n)}\sum_{k=0}^{\infty}\frac{(-1)^k}{k!}
\int_{\RRR^{\n k}} dy_1\cdots dy_k \prod_{j=1}^k\Big(1-e^{-\b\f(q_1-y_j)}\Big)
\r_{n-1+k}(q_2,\cdots,q_n,y_1,\cdots,y_k) \nn\\
&&\ \ \ \ \ + \EE_{n,q_0}(q_1,\cdots,q_n)\;, \label{eq:ultima}
\eea
where the term $k=0$ in the first sum has to be interpreted as $1,$ and the error term is given by
\bea
&& \EE_{n,q_0}(q_1, \cdots, q_n) = e^{-\b W_{q_0}(q_2,\cdots,q_n)}\sum_{k=0}^{\infty}\frac{(-1)^k}{k!}
\int_{\RRR^{\n k}} dy_1\cdots dy_k \prod_{j=1}^k\Big(1-e^{-\b\f(y_j)}\Big)\nn\\
&&\ \ \ \ \ \cdot\Big[ \r_{n}(q_1,q_2,\cdots,q_n) \r_k(q_0+y_1,\cdots,q_0+y_k) \nn\\
&&\ \ \ \ \ - e^{\b W_{q_1}(q_0+y_1,\cdots,q_0+y_k)}\r_{n+k}(q_1,q_2,\cdots,q_n,q_0+y_1,\cdots,q_0+y_k)\Big]\nn\\
&& + \r_n(q_1,\cdots,q_n) \left(1-e^{-\b W_{q_0}(q_2,\cdots,q_n)}\right) \nn\\
&&\ \ \ \ \ \cdot\sum_{k=0}^{\infty}\frac{(-1)^k}{k!} \int_{\RRR^{\n k}} dy_1\cdots dy_k \prod_{j=1}^k\Big(1-e^{-\b\f(y_j)}\Big)
\r_{k}(y_1,\cdots,y_k)\nn\\
&&+e^{-\b W_{q_1}(q_2,\cdots,q_n)}\sum_{k=0}^{\infty}\frac{(-1)^k}{k!}
\int_{\RRR^{\n k}} dy_1\cdots dy_k \prod_{j=1}^k\Big(1-e^{-\b\f(q_1-y_j)}\Big)\nn\\
&&\ \ \ \ \ \cdot \left[ e^{\b W_{q_0}(y_1,\cdots,y_k)}\r_{n+k}(q_0,q_2,\cdots,q_n,y_1,\cdots,y_k)
- \r \r_{n-1+k}(q_2,\cdots,q_n,y_1,\cdots,y_k)\right]\;. \label{eq:sommaR}
\eea
The cluster property (\ref{eq:DIS}) and the Dominated Convergence Theorem 
(which can be applied by assumption (\ref{eq:rnbound})) imply that $\EE_{n,q_0} \longrightarrow 0$ as 
$|q_0|\longrightarrow 0.$

The sum in the left hand side of (\ref{eq:ultima}) is a strictly positive constant depending on $\b$ and 
$\r_k, k\geq 1;$ this follows by using that $\r_k$ are correlation functions of a probability measure, and it is
checked for completeness in Appendix A. The direct statement of the Theorem is thus proved by calling
\bea
z = \frac{\r}{\Big[1+\sum_{k=1}^{\infty}\frac{(-1)^k}{k!}
\int_{\RRR^{\n k}} dy_1\cdots dy_k \prod_{j=1}^k\Big(1-e^{-\b\f(y_j)}\Big)\r_{k}(y_1,\cdots,y_k)\Big]}\;.\label{eq:z}
\eea
\end{proof}

\bigskip

Clearly, the direct statement of Theorem \ref{thm:main} has no meaning for {\em all} values of the parameters $\r,\b,$ since
it could happen that, for given values of those parameters, there are no solutions to the Kirkwood--Salsburg 
equations obeying the hypotheses of the Theorem. In particular, we refer to translation invariance and cluster 
properties, which are only proved to be valid inside the ``gas phase region'' (small $\r\ $ and small $\b$). 
We want to stress also that, outside that region, there could be multiple--valued solutions to Eq. (\ref{eq:KS}), including 
both gaseous and liquid states. Existence and uniqueness are assured by the Theorem just for $\xi$ small, 
as explained by the next

\begin{corollario} \label{cor:main}
In the hypotheses of Theorem \ref{thm:main}, for $\xi$ sufficiently small 
the state is uniquely determined by $\r, \b,$ and it coincides with the (unique) solution of (\ref{eq:KS}).
\end{corollario}

\begin{proof}[Proof of Corollary \ref{cor:main}]
The result follows from the well known theory of convergence of the Mayer expansion for 
$z$ small \cite{Pe63a}, \cite{Ru63}, after noting from Eq. (\ref{eq:z}) that $z = O(\xi)$ for $\xi$ small. 
We sketch the proof for completeness.

By iteration of (\ref{eq:KS}) we get the formal expansions
\bea
&&\r=z \sum_{p=0}^{\infty}c_{1,p}z^p\nn\\
&&\r_n(q_1,\cdots,q_n) = z\sum_{p=0}^{\infty}c_{n,p}(q_1,\cdots,q_n)z^p\;,\ \ \ \ \ n>1\;, \label{eq:Mayercf}
\eea 
where the coefficients are defined in terms of $\b$ and $\f$ by the explicit recursive relation
\bea
&&c_{n,0} = \d_{n,1} \label{eq:Mayer} \\
&&c_{n,p+1}(q_1,\cdots,q_n) = e^{-\b W_{q_1}(q_2,\cdots,q_n)}
\Big[ \d_{n>1}c_{n-1,p}(q_2,\cdots,q_n)\nn\\
&&\ \ \ \ \ + \sum_{k=1}^{\infty}\frac{(-1)^k}{k!}\int_{\RRR^{\n k}} dy_1\cdots dy_k
\prod_{j=1}^k \Big(1-e^{-\b\f(q_1-y_j)}\Big)c_{n-1+k,p}(q_2,\cdots,q_n,y_1,\cdots,y_k)
\Big]\;, \nn
\eea
for $p\geq 0.$ In particular, it follows that $c_{n,n-1}(q_1,\cdots,q_n) = e^{-\b\sum_{i<j}^{0,n} \f(q_i-q_j)}$ 
and $c_{n,p}(q_1,\cdots,q_n) = 0$ for $p < n-1$.  

Defining
\bea
I_{\b} = \int_{\RRR^{\n}}\left|1-e^{-\b\f(x)}\right|dx\;, \label{eq:Ib}
\eea
by induction on $p$ and using the stability of the potential (which implies, for any configuration of $n$ particles,
the existence of $i\in(1,\cdots,n)$ such that $W_{q_i}(q_1,\cdots,q_{i-1},q_{i+1},\cdots,q_n) \geq -2B$) 
one finds the following estimate uniform in $q_1,\cdots,q_n\in\RRR^\n:$
\be
|c_{n,p}(q_1,\cdots,q_n)| \leq I_{\b}^{-(n-1)} (I_{\b}e^{1+2\b B})^p\;.
\ee
Hence the expansions (\ref{eq:Mayercf}) are absolutely convergent uniformly in the coordinates,
as soon as 
\be
|z| < (I_{\b}e^{1+2\b B})^{-1}\;. \label{eq:zbound}
\ee
The first equation ($n=1$) is, in this case, the expansion of $\r$ in powers of $z,$ and it is of the 
form $\r = z + O(z^2):$ thus it can be inverted for $z$ small, to determine $z$ as a function of $\r$ and $\b.$

Therefore, to obtain the corollary it is sufficient to take
\be
\xi < (2I_{\b}e^{1+2\b B})^{-1}\;. \label{eq:xibound}
\ee
In fact, with this choice the denominator in (\ref{eq:z}) is bounded from below by 
$1-I_{\b}\xi e^{2\b B} e^{I_{\b}\xi} > 1/2,$ so that $|z| \leq 2\xi$ and Eq. (\ref{eq:zbound}) is satisfied.
\end{proof}


\bigskip

The proof of Theorem \ref{thm:main} can be easily adapted to cover the more general case of a non invariant
state for which translation invariance holds just as a boundary condition at infinity: 
that is the case of a system of particles in an infinite container. 
Consider an open unbounded set $\L_{\infty}\subset \RRR^\n$ with a smooth boundary $\partial\L_{\infty}$
and satisfying the following properties:

{\em (a)} $\L_{\infty}$ is polygonally connected;

{\em (b)} for any $q \in \L_{\infty},$ there exists a polygonal path $\G(q)$ connecting $q$ to $\infty$
such that 
\be
\dist(\partial\L_{\infty}, \{y \in \G(q) \mbox{ s.t. }|y|>n\}) \longrightarrow +\infty
\ee
as $n\longrightarrow+\infty$ (here $\dist$ is the usual distance between sets in $\RRR^\n$).
That describes a class of reasonable geometries for infinite containers of particles.

In the following we shall assume, for simplicity, reflecting boundaries. 
Notice that the Maxwellian assumption ensures that the correlation functions take the same value on 
configurations that correspond to the incoming and outcoming state of an elastic collision particle--wall. 
Moreover, as we will see, the value of the correlation functions on the boundary $\partial\L_{\infty}$ 
does not play any role in the proof below: different kinds of boundary condition could be also treated,
such as walls modeled by a smooth external potential.

The phase space associated to the system, 
denoted $\HH_{\L_{\infty}},$ is defined as in point 1) of Section \ref{sec:setup} with the $q_i$
restricted to $\L_{\infty}.$ All the other definitions of Section \ref{sec:setup} are extended as well to the system on $\L_{\infty},$
just by restricting the coordinates $q_i\in\L_{\infty}.$ In particular, a smooth Maxwellian state on $\HH_{\L_{\infty}}$ is a collection 
of probability measures on $\HH_{\L}$ with $\L\subset\L_{\infty}$ bounded open, satisfying properties {\bf a.}--{\bf e.} of Section 
{\ref{sec:setup}}, having correlation functions of the form ({\ref{eq:Maxwell}}) with 
$g_n\in C^1(\L_{\infty}^n)\cap C(\ol\L_{\infty}^n)$ (the bar indicates closure in the usual topology)
and satisfying the estimates (\ref{eq:rnbound}) and (\ref{eq:gradrnbound}) over $\L_{\infty}^n.$

The stationary BBGKY hierarchy of equations for such a state reduces to 
\bea
&&\nabla_{q_1} \r_n(q_1,\cdots,q_n)=-\b\Big[\nabla_{q_1}W_{q_1}(q_2,\cdots,q_n)\r_n(q_1,\cdots,q_n)\nn\\
&&\ \ \ \ \ \ \ \ \ \ \ \ \ \ \ \ \ \ \ \ \ + \int_{\L_{\infty}}dy \nabla_{q_1} \f(q_1-y)\r_{n+1}(q_1,\cdots,q_n,y)\Big]\;,
\ \ \ \ \ \ \ \ \ \ n \geq 1\;,\label{eq:grad-rhoinfcont}
\eea
for $q_1,\cdots,q_n\in\L_{\infty},$ which we shall integrate with the boundary conditions:

{\em (i)} $\r_n$ satisfy the cluster property (\ref{eq:DIS}) on $\L_{\infty}$ as soon as $\dist(B_m,\partial\L_{\infty})
\longrightarrow +\infty;$

{\em (ii)} $\r_n$ satisfy the following property which we call {\em invariance at infinity}: there exists a sequence of 
translation invariant functions $\{f_n\}_{n=1}^{\infty}, f_n:\RRR^{\n n}\rightarrow\RRR^+$ (being the correlation functions
of some state on $\HH$), a constant $\tilde C>0,$ and two monotonous decreasing functions $\tilde u(\cdot), \e(\cdot)$ 
vanishing at infinity, such that
\bea
&&|\r_n(q_0+q_1,q_0+q_2,\cdots,q_0+q_n) - f_n(q_1,q_2,\cdots,q_n)|\nn\\
&&\ \ \ \ \ \ \ \ \ \ \ \ \ \ \ \ \ \ \ \ \ \ \ \ \ \leq {\tilde C}^n \left[\tilde u(|q_0|) +
\e\left(\dist(\partial\L_{\infty}, \{q_0+q_1,\cdots,q_0+q_n\})\right)\right] \label{eq:invinfprop}
\eea
for all $q_0,q_1,\cdots,q_n$ with $q_0+q_1,\cdots,q_0+q_n\in\L_{\infty}.$ 

We will name $\r\equiv f_1(q_1).$ The following extension of Theorem \ref{thm:main} holds:
\begin{teorema} \label{thm:infcont}
If a smooth Maxwellian state on $\HH_{\L_{\infty}}$ is a stationary solution of the BBGKY hierarchy satisfying cluster
boundary conditions and invariance at infinity, then there exists a constant $z$ such that the correlation functions of the
state satisfy
\bea
&&\r_n(q_1,\cdots,q_n)= z e^{-\b W_{q_1}(q_2,\cdots,q_n)}\Big[\r_{n-1} (q_2,\cdots,q_n)\label{eq:KSinfcont}\\
&&\ \ \ \ \ \ \ \ \ \ +\sum_{m=1}^{\infty}\frac{(-1)^m}{m!}\int_{\L^{m}_{\infty}} dy_1 \cdots dy_m 
\prod_{j=1}^m \left(1-e^{-\b\f(q_1-y_j)}\right) \r_{n-1+m} (q_2,\cdots,q_n,y_1,\cdots,y_m)\Big]\;.\nn
\eea
Conversely, a smooth Maxwellian state on $\HH_{\L_{\infty}}$ with parameter $\b$
and satisfying (\ref{eq:KSinfcont}) is a stationary solution of the BBGKY hierarchy.
\end{teorema}

These are the Kirkwood--Salsburg equations in the infinite container. For $n=1$ the first term in the right hand side
has to be interpreted as $z.$ In this case the equation is not independent on $q_1:$ a definition of $z$ in terms of explicitly
constant functions follows using (\ref{eq:invinfprop}), by sending $q_1$ to infinity in such a way that 
$\dist(q_1,\partial\L_{\infty})\longrightarrow +\infty,$ which is certainly possible in our assumption
{\em (b)} on the geometry of the container (see Eq. (\ref{eq:zinfcont}) below).

\begin{proof}[Proof of Theorem \ref{thm:infcont}] 
All that is said in the proof of Theorem \ref{thm:main} up to the formula (\ref{eq:KSq0}) can be repeated
here by restricting the coordinates to $\L_{\infty},$ substituting the integration region $\RRR^\n$ with $\L_{\infty},$
and the straight line $\overrightarrow{q_0q_1}$ with a polygonal path entirely contained in $\L_{\infty}$ connecting $q_0$ 
to $q_1.$ We obtain that Eq. (\ref{eq:KSq0}) is valid with $q_0\in\L_{\infty}$ and the integrals restricted to $\L_{\infty}^k.$
Take the limit of this expression as $|q_0|\rightarrow\infty$ with $q_0$ moving along a path $\G(q_1)$ defined as in point {\em (b)}
above: properties {\em (i)} and {\em (ii)} then imply 
\bea
&&\r_n(q_1,\cdots,q_n) \Big[1+\sum_{k=1}^{\infty}\frac{(-1)^k}{k!}
\int_{\RRR^{\n k}} dy_1\cdots dy_k \prod_{j=1}^k\Big(1-e^{-\b\f(y_j)}\Big)f_{k}(y_1,\cdots,y_k)\Big]\\
&&=\r e^{-\b W_{q_1}(q_2,\cdots,q_n)}\sum_{k=0}^{\infty}\frac{(-1)^k}{k!}
\int_{\L_{\infty}^k} dy_1\cdots dy_k \prod_{j=1}^k\Big(1-e^{-\b\f(q_1-y_j)}\Big)
\r_{n-1+k}(q_2,\cdots,q_n,y_1,\cdots,y_k)\;; \nn
\eea
this can be shown via a dominated convergence argument as after formula (\ref{eq:KSq0}) in Theorem \ref{thm:main}.
The factor in the square brackets on the left hand side is a strictly positive constant depending on $\b$ and 
$f_k, k\geq 1$ (apply the discussion in Appendix A). The direct statement of the Theorem is thus proved by calling
\bea
z = \frac{\r}{\Big[1+\sum_{k=1}^{\infty}\frac{(-1)^k}{k!}
\int_{\RRR^{\n k}} dy_1\cdots dy_k \prod_{j=1}^k\Big(1-e^{-\b\f(y_j)}\Big)f_{k}(y_1,\cdots,y_k)\Big]}\;.\label{eq:zinfcont}
\eea
\end{proof}

The proof of Corollary \ref{cor:main} can be also adapted to the present case in a straightforward way. Relations 
(\ref{eq:Mayercf})--(\ref{eq:Mayer}) with the constants $\r, c_{1,p}$ replaced by functions $\r_1(q_1), c_{1,p}(q_1)$ 
and all the involved coordinates
restricted to $\L_{\infty},$ show that if $\xi$ is taken as in (\ref{eq:xibound}), Equations (\ref{eq:KSinfcont}) have a unique 
solution determined by $\r = f_1$ and $\b.$ Notice that this solution is actually coincident with that of the infinite 
space Kirkwood--Salsburg equations (\ref{eq:KS}), in the limit of infinite distance of the coordinates 
from $\partial\L_{\infty}.$
The activity $z$ in $\L_{\infty}$ is given by the inversion of the power series
\be
\r=z \sum_{p=0}^{\infty}
\left(\underset{\underset{\dist(\partial\L_{\infty},q)\rightarrow\infty}{|q|\rightarrow\infty}}{\lim} c_{1,p}(q)\right)z^p\;,
\ee
where the coefficients of the expansion are independent on the way the limit is taken.

\begin{corollario} \label{cor:infcont}
In the hypotheses of Theorem \ref{thm:infcont}, for $\xi$ sufficiently small 
the state is uniquely determined by $\r, \b,$ and it coincides with the (unique) solution of (\ref{eq:KSinfcont}). $\hfill\Box$
\end{corollario}


\section{The hard core limit} \label{sec:hc}

In this section we deal with the infinite system of hard core particles with diameter $d>0.$
The Hamiltonian $H(X)$ is defined by $(\ref{eq:Ham})$ with
\be
\f = \f_d(q)=
\left\{\begin{array}{ccc}
\infty, & &  |q| < d\\ 0, & & |q| \geq d
\end{array}\right.\;.
\ee
Definitions of Section \ref{sec:setup} can be easily extended: see \cite{AGGLM}.
In particular, the phase space is $\HH_d =\Big\{X = \{x_i\}_{i=1}^{\infty} = \{(q_i,p_i)\}_{i=1}^{\infty},
x_i\in\RRR^{\nu}\times\RRR^{\nu} \ \Big|\  |q_i-q_j|\geq d \mbox{ for } i\neq j\Big\}$\footnote{
Here we are considering for simplicity pure hard core systems in infinite space; an additional 
potential of the class introduced in Section \ref{sec:setup} (but also possibly singular in $|q|=d$)
can be added in the discussion of the present section in an obvious way, and more general geometries 
can be considered along the lines of Theorem \ref{thm:infcont} and the discussion thereof.},
while a state $\mu$ is a probability measure on the Borel sets of $\HH_d$ having correlation
functions 
\be
\ol\r_n: \HH_d^{(n)} \rightarrow \RRR^+\;,
\ee
where
\be
\HH_d^{(n)} = \Big\{\{x_i\}_{i=0}^n = \{(q_i,p_i)\}_{i=0}^n,
x_i\in\RRR^{\nu}\times\RRR^{\nu} \ \Big|\  |q_i-q_j| > d \mbox{ for } i\neq j\Big\}\;. 
\ee
The state $\mu$ is called smooth Maxwellian if its correlation functions are as in the first line of Eq. (\ref{eq:Maxwell}), 
they are $C(\ol\HH_d^{(n)})$ and piecewise $C^1(\HH_d^{(n)}),$
and satisfy bounds as in (\ref{eq:rnbound}), (\ref{eq:gradrnbound}) (without the exponentials).

From now on we abbreviate
\be
\RRR^{n\nu}_d = \Big\{(q_1,\cdots,q_n)\in\RRR^{\nu n}\ \Big|\ |q_i-q_j| > d \mbox{ for } i\neq j\Big\}\;.
\ee
Moreover, we let
\be
\O_i(q_1,\cdots,q_n) = \Big\{\o \in S^{\nu-1} \ \Big|\ |q_i+dw-q_j| > d \mbox{ for every $j\in(1,\cdots,n), j\neq i$}\Big\}\;.
\label{eq:Oi}
\ee
Following \cite{Cerc75} (see also \cite{IP} for a more rigorous treatment), 
we say that the smooth Maxwellian state is a stationary solution of the {\em hard core BBGKY hierarchy} if, for $n\geq 1,$
its spatial correlation functions $\r_n:\RRR^{\nu n}_d \rightarrow\RRR^+, \r_n\in C(\ol\RRR^{\nu n}_d)$ 
and piecewise $C^1(\RRR^{\nu n}_d),$ satisfy
\be
\nabla_{q_1}\r_n(q_1,\cdots,q_n) 
= - d^{\nu-1} \int_{\O_1(q_1,\cdots,q_n)}d\o \o \r_{n+1}(q_1,\cdots,q_n,q_1+d \o)\;,\label{eq:hcgr}
\ee
where $d\o$ denotes the surface element on the unit sphere (for $\nu =1$ the integral reduces to $\sum_{\o=\pm 1}$). 
This is equivalent to say that, for $n \geq 1,$
\bea
&&\sum_{i=1}^n p_i\cdot \nabla_{q_i}\ol\r_n(x_1,\cdots,x_n) \nn\\
&&= - d^{\nu-1} \sum_{i=1}^n \int_{\O_i(q_1,\cdots,q_n)\times \RRR^{\nu}}d\o d\pi\  \o \cdot (p_i - \pi) 
\ol\r_{n+1}(x_1,\cdots,x_n,q_i+d \o,\pi) \label{eq:BBGKYhc}
\eea
for any $(x_1,\cdots,x_n)\in\HH_d^{(n)},$ where $\o$ varies over $\O_i(q_1,\cdots,q_n)$ and 
$\pi$ varies in $\RRR^\nu\ $\footnote{These equations, as derived in \cite{Cerc75}, should be complemented 
with the boundary conditions imposing that the correlation functions take the same value on configurations that 
correspond to the incoming and outcoming state of a collision; which of course is guaranteed by the Maxwellian assumption.}.
Observe that in the hard core case, if the state is also invariant, the equations are parametrized by only one positive
(potential--independent) constant, $\r \equiv \r_1(q_1);$ i.e. $\b$ does not appear. 
Finally, notice that the cluster property can be formulated as in (\ref{eq:DIS}) for
the functions defined on $\RRR^{\nu (n+m)}_d.$

The direct integration procedure established in the previous section cannot be applied to solve
the hierarchy (\ref{eq:hcgr}), as already stressed in the Remark on page \pageref{thm:main}, the difficulty coming from the presence of ``holes'' in the phase space.
However, from the results stated in Sec. \ref{sec:mr} it follows that the solution of (\ref{eq:hcgr}) describing the 
equilibrium correlation functions of the hard core system, 
defined (uniquely for $\r$ small) via its corresponding Kirkwood--Salsburg equations,
can be approximated with solutions of the smooth hierarchies (\ref{eq:grad-rho}),
with few restrictions on the form of the regular potentials that can be used. 

More precisely, let $\f^{(\e)}\in C^{1}(\RRR^{\n}), \e>0,$ be any family of radial, positive
potentials with compact support, converging pointwise to the hard core potential:
\be
\f^{(\e)}(q) \underset{\e\to 0}{\longrightarrow} \f_d(q),\ \ \ \ \ \mbox{for }|q| \neq d\;. \label{eq:fe}
\ee
Denote $\MM^{(\b,\xi)}, \b,\xi>0,$ the set of all smooth Maxwellian and invariant states on $\HH$ 
with parameter $\b,$ and spatial correlation functions of class $C^1$ in its variables, 
obeying estimates of the form~$\r_n\leq\xi^n, |\nabla_{q_i}\r_n|\leq C_n\xi^n$ 
(with possibly different constants $C_n$) and satisfying the cluster property (\ref{eq:DIS}).
Indicate $B_d(q)$ the ball with radius $d$ and center $q\in\RRR^{\n},$ and
\be
\RRR^{m\n}_d(q_1,\cdots,q_n) = \left\{(y_1,\cdots,y_m)\in\RRR^{\n m}_d\ \Big|\ (q_1,\cdots,q_n,y_1,\cdots,y_m) 
\in \RRR^{\n(n+m)}_d\right\}\;.
\ee
Then the following holds:

\begin{teorema} \label{prop:HC}
Fix $\b>0,$ and sufficiently small $\r>0.$ Then there exists a (small) constant $\xi>\r$ such that:

\ni (i) for any $\e>0$ there is a unique state in $\MM^{(\b,\xi)}$ with spatial correlation functions 
$\{\r_n^{(\e)}\}_{n=1}^{\infty}$ solving the hierarchy (\ref{eq:grad-rho}) 
with potential $\f^{(\e)},$ and $\r_1^{(\e)}(q_1)\equiv\r;$

\ni (ii) it is $|\r_n^{(\e)}(q_1,\cdots,q_n)| \leq (2\xi)^n e^{-\b \sum_{i\neq j}\f^{(\e)}(q_i-q_j)}$ uniformly in $\e$, $n\geq1$ and $(q_1,\cdots,q_n)\in\RRR^{\n n};$ moreover,
\be
\r_n^{(\e)}(q_1,\cdots,q_n)  \underset{\e\to 0}{\longrightarrow} \r_n(q_1,\cdots,q_n)
\ee
uniformly in every compact subset of $\RRR^{\n n}_d,$ where the functions 
$\r_n:\RRR^{\n n}_d\rightarrow\RRR^+$ are given by the hard core Kirkwood--Salsburg equations:
\bea
&&\r_n(q_1,\cdots,q_n)= z \Big[\r_{n-1} (q_2,\cdots,q_n)\label{eq:hcKS}\\
&&\ \ \ \ \ \ \ \ +\sum_{m=1}^{\infty}\frac{(-1)^m}{m!}\int_{\RRR^{m\n}_d(q_2,\cdots,q_n)\bigcap (B_d(q_1))^m} 
dy_1 \cdots dy_m \r_{n-1+m} (q_2,\cdots,q_n,y_1,\cdots,y_m)\Big]\;;\nn
\eea
(iii) the limit functions $\r_n$ satisfy the hard core hierarchy (\ref{eq:hcgr}).
\end{teorema}

Notice that the sum in the right hand side of (\ref{eq:hcKS}) is finite, because of the hard core exclusion.
From points {\em (ii)} and {\em (iii)} of the theorem, and the known theory of equations
(\ref{eq:hcKS}) for small densities, it follows that the limit functions $\r_n$ provide a smooth Maxwellian 
invariant state on $\HH_d$ which is a solution of the stationary hard core BBGKY hierarchy with cluster 
boundary conditions.

It would be interesting to understand if it is possible to work out an iterative procedure that 
integrates Eq. (\ref{eq:hcgr}) directly (without assuming boundary conditions), as we are able 
to do in the smooth case. This is, as far as we know, an open problem.
A direct integration can be carried out for $\r$ small in the case $\nu = 1:$ we discuss 
it in Appendix B.  

\bigskip

\begin{proof}[Proof of Theorem \ref{prop:HC}]
Applying the direct statement of Theorem \ref{thm:main}, we have that any state in $\MM^{(\b,\xi)}$
with fixed density $\r<\xi,$ solving the stationary BBGKY hierarchy with interaction $\f^{(\e)},$
satisfies also Eq. (\ref{eq:KS}) with the same interaction, for some value of the activity $z_\e.$ 
By the proof of Corollary \ref{cor:main}, this
last set of equations has a unique solution if $\xi$ is taken as in (\ref{eq:xibound}). Thus point {\em (i)}
follows by chosing $\xi$ (hence $\r$) in such a way that
\be
\xi < \frac{1}{2e \sup_{\e>0} \int_{\RRR^{\n}}\left(1-e^{-\b\f^{(\e)}(x)}\right)dx}\;. \label{eq:boundxie}
\ee

The solution $\r_n^{(\e)}$ for given $\e$ can be expanded in absolutely convergent power series of the activity, 
so that we have formula (\ref{eq:Mayercf}) with $z$ replaced by $z_\e,$ a superscript $(\e)$
added to $\r_n$ and coefficients of the expansions $c_{1,p}^{(\e)}, c_{n,p}^{(\e)}$
defined by Eq. (\ref{eq:Mayer}) with potential $\f^{(\e)}.$ 
Since $\f^{(\e)}$ is positive, it follows
(see \cite{Gr62}) that the coefficients of the series expansions have alternating signs, and that the 
same expansions have the {\em alternating bound property} \cite{Pe63b}, which means in particular
that, for $z_\e>0$ (which is certainly true if $\r$ is small enough for all $\e>0$),
they can be bounded with their leading terms as:
\be
\r_n^{(\e)} < z_{\e}^n c_{n,n-1} = z_{\e}^n e^{-\b\sum_{i<j}^{0,n} \f^{(\e)}(q_i-q_j)}\;. \label{eq:Groen}
\ee
This, together with (\ref{eq:boundxie}) and (\ref{eq:z}), gives the estimate of point {\em (ii)} of the theorem.

Assuming by induction on $p$ that $c_{n,p}^{(\e)}\rightarrow c_{n,p}^{(0)}$
as $\e\rightarrow 0,$ where $c_{n,p}^{(0)}$ are the coefficients of the formal expansion obtained by iteration
of Eq. (\ref{eq:hcKS}), we obtain from (\ref{eq:Mayer}) that for $|q_i-q_j| > d$:
\bea
&&\lim_{\e\rightarrow 0}c_{n,p+1}^{(\e)}(q_1,\cdots,q_n)
=  \Big[ \d_{n>1}c_{n-1,p}^{(0)}(q_2,\cdots,q_n)\nn\\
&&\ \ \ \ \ \ \ \ + \sum_{k=1}^{\infty}\frac{(-1)^k}{k!}
\int_{\RRR^{k\n}_d(q_2,\cdots,q_n)\bigcap (B_d(q_1))^m} dy_1\cdots dy_k
c_{n-1+k,p}^{(0)}(q_2,\cdots,q_n,y_1,\cdots,y_k)\Big]\nn\\
&&\ \ \ \ \ \equiv c_{n,p+1}^{(0)}(q_1,\cdots,q_n)\;.
\eea
This ends the proof of point {\em (ii)}.

Point {\em (iii)} is now a particular case of the following Lemma, which is the analogous of the 
converse statement of Theorem \ref{thm:main}:
\begin{lemma} \label{lem:converse}
If a smooth Maxwellian state on $\HH_d$ satisfies (\ref{eq:hcKS}), then it is a stationary solution of the hard core
BBGKY.
\end{lemma}

{\em Proof of Lemma \ref{lem:converse}.}
We compute the gradient with respect to $q_1$
of expression (\ref{eq:hcKS}), in a configuration $(q_1,\cdots,q_n)\in\RRR^{\n n}_d.$
Remind that the series in the right hand side is actually a finite sum. We have
\bea
&&\nabla_{q_1}\r_n(q_1,\cdots,q_n) \\
&&=z\sum_{k=1}^{\infty}\frac{(-1)^k}{k!}\nabla_{q_1}
\int_{\RRR^{k\n}_d(q_2,\cdots,q_n)\bigcap (B_d(q_1))^k} dy_1 \cdots dy_k 
\r_{n-1+k} (q_2,\cdots,q_n,y_1,\cdots,y_k)\nn\\
&&= z\sum_{k=1}^{\infty}\frac{(-1)^k}{k!}k
\int_{\RRR^{(k-1)\n}_d(q_2,\cdots,q_n)\bigcap (B_d(q_1))^{k-1}} dy_1 \cdots dy_{k-1} \nn\\
&&\ \ \ \ \ \ \ \ \ \ \cdot\nabla_{q_1}\int_{\RRR^{\n}_d(q_2,\cdots,q_n,y_1,\cdots,y_{k-1})\bigcap B_d(q_1)} 
dy^* \r_{n-1+k} (q_2,\cdots,q_n,y^*,y_1,\cdots,y_{k-1})\;, \nn
\eea
where the second equivalence holds by symmetry in the exchange of particle labels and by uniform
convergence of the integrals. 
The integral in the last line of the formula is extended on a region which has positive volume 
for sufficiently small $k,$ and piecewise smooth boundary: the ball centered in $q_1$ minus 
the union of the balls centered in the points $q_2,\cdots,q_n,y_1,\cdots,y_{k-1}.$ 
Being the integrand function $\r_{n-1+k}(\cdots,y^*,\cdots,)$ continuous in the closure of its 
domain, it is easy to see that the gradient with respect to $q_1$ of such an integral is given by 
the surface integral of the restriction of the function over the part of the boundary of $B_d(q_1)$
that remains ouside the other balls, i.e. using the notations of (\ref{eq:Oi}) and (\ref{eq:hcgr}):
\be
d^{\nu-1} \int_{\O_1(q_1,\cdots,q_n,y_1,\cdots,y_{k-1})}d\o \o \r_{n-1+k}
(q_2,\cdots,q_n,q_1+d \o,y_1,\cdots,y_{k-1})\;.
\ee
Interchanging the integrations we find
\bea
&&\nabla_{q_1}\r_n(q_1,\cdots,q_n) \\
&&= d^{\nu-1} \int_{\O_1(q_1,\cdots,q_n)}d\o \o z\sum_{k=1}^{\infty}\frac{(-1)^k}{k!}k\nn\\
&&\ \ \ \ \ \cdot \int_{\RRR^{(k-1)\n}_d(q_2,\cdots,q_n,q_1+d\o)
\bigcap (B_d(q_1))^{k-1}} dy_1 \cdots dy_{k-1} \r_{n-1+k}(q_2,\cdots,q_n,q_1+d \o,y_1,\cdots,y_{k-1})\;.\nn
\eea
Using that (\ref{eq:hcKS}) holds also, by continuity, over the boundary of $\RRR^{\n n}_d,$
we recognize the function 
\be
- \r_{n+1}(q_1,q_2,\cdots,q_n,q_1+d\o)
\ee
in the above expression, so obtaining Eq. (\ref{eq:hcgr}). 
\end{proof}


\section{Conclusions} \label{sec:conc}

We have studied the stationary BBGKY hierarchy of equations for infinite classical systems of
particles, assuming the usual Gaussian distribution of momenta.
We proved equivalence with the set of Kirkwood--Salsburg equations through
a constructive iterative method of integration. We extended the result of \cite{GV75}
to a larger class of potentials and states. We have stated partial results on the hard core hierarchy:
it would be interesting to understand if it is possible to work out a method of integration that 
can be applied directly to this last case.

We hope that the methods developed in this paper can help to handle different types of boundary 
condition and to understand different models of hierarchy of equations, such as those that can be
obtained replacing the Maxwellian assumption with a more general local equilibrium assumption
in which $\b$ and $z$ are position--dependent.
Notice that the formula Eq. (\ref{eq:KSq0}) in the proof of the main Theorem 
(derived without the use of properties of the state at infinity) has the remarkable form of a  
symmetric generalization of the Kirkwood--Salsburg equations, in which a free parameter $q_0$ is added. 
As already shown by Theorem \ref{thm:infcont}, this degree of freedom
may be used to impose different kinds of boundary condition. For instance, it would be of interest 
the study of the infinite hierarchy for a system of particles in a finite box: in this case even the equilibrium
problem has to be treated.

\bigskip

\ni{\bf Acknowledgments\\}
\ni The authors are grateful to Giovanni Gallavotti for suggesting the problem, for stimulating discussions
and encouragement. They also thank Alessandro Giuliani and Mario Pulvirenti for many useful discussions.

\bigskip


\section*{Appendix A. Positivity of the activity} \label{app:z}

In this appendix we check that the constant introduced by (\ref{eq:z}) in the proof of Theorem \ref{thm:main} 
is well defined and positive. We put $x_j = (y_j,p_j).$
By assumption (\ref{eq:Maxwell}) the denominator in (\ref{eq:z}) can be written as
\bea
&&1+\sum_{k=1}^{\infty}\frac{(-1)^k}{k!}
\int_{(B_R\times\RRR^\n)^k} dx_1\cdots dx_k \prod_{j=1}^k\Big(1-e^{-\b\f(y_j)}\Big)\ol\r_{k}(x_1,\cdots,x_k)\nn\\
&&\ \ \ + \sum_{k=1}^{\infty}\frac{(-1)^k}{k!}
\int_{(\RRR^\n\setminus B_R)^k} dy_1\cdots dy_k \prod_{j=1}^k\Big(1-e^{-\b\f(y_j)}\Big)\r_{k}(y_1,\cdots,y_k)\;,\label{eq:zexptot}
\eea
where $B_R$ is the ball centered in $0$ and with radius $R>0.$
Using the definition of correlation functions, Eq. (\ref{eq:defcf}), we can easily rewrite the first line as
\bea
&& 1+ \sum_{k=1}^{\infty}\sum_{p=1}^{k}\frac{(-1)^p}{p!(k-p)!}
\int_{(B_R\times\RRR^\n)^k} dx_1\cdots dx_k \prod_{j=1}^p\Big(1-e^{-\b\f(y_j)}\Big)\m_{B_R}^{(k)}(x_1,\cdots,x_k)\;.
\label{eq:poszexp}
\eea
Expanding the product, the integral in this expression is
\bea
&&\sum_{n=0}^p (-1)^n \int_{(B_R\times\RRR^\n)^k} dx_1\cdots dx_k \sum_{1\leq j_1<\cdots <j_n\leq p}
\left(\prod_{i=1}^n e^{-\b\f(y_{j_i})} \right)\m_{B_R}^{(k)}(x_1,\cdots,x_k)\nn\\
&& = \sum_{n=0}^p (-1)^n \binom{p}{n} C^{(k,n)}_{B_R}\;, \label{eq:poszint}
\eea
where the equality holds by symmetry of $\m_{B_R}^{(k)},$ with
\be
C^{(k,n)}_{B_R} := \int_{(B_R\times\RRR^\n)^k} dx_1\cdots dx_k 
\left(\prod_{i=1}^n e^{-\b\f(y_{i})} \right)\m_{B_R}^{(k)}(x_1,\cdots,x_k)\;.
\ee
Putting (\ref{eq:poszint}) into (\ref{eq:poszexp}) and interchanging the sums, we have
\bea
&&1+\sum_{k=1}^{\infty}\frac{1}{k!}\sum_{n=0}^k\frac{(-1)^n}{n!} C^{(k,n)}_{B_R}
\sum_{p=n}^{k}(-1)^p(1-\d_{p,0})\frac{k!}{(k-p)!(p-n)!}\nn\\
&& =1+\sum_{k=1}^{\infty}\frac{1}{k!}\sum_{n=0}^k\frac{1}{n!} C^{(k,n)}_{B_R}
\sum_{p=0}^{k-n}(-1)^p(1-\d_{p,0}\d_{n,0})\frac{k!}{(k-n-p)!p!} \nn\\
&& =1+\sum_{k=1}^{\infty}\frac{1}{k!}\sum_{n=1}^k\frac{1}{n!} C^{(k,n)}_{B_R}
k(k-1)\cdots(k-n+1)\d_{n,k}-\sum_{k=1}^{\infty}\frac{1}{k!}C^{(k,0)}_{B_R}\nn\\
&& =1+\sum_{k=1}^{\infty}\frac{1}{k!}C^{(k,k)}_{B_R}-
\sum_{k=1}^{\infty}\frac{1}{k!}C^{(k,0)}_{B_R}\nn\\
&& = \sum_{k=0}^{\infty}\frac{1}{k!}C^{(k,k)}_{B_R}\;,
\eea
having used the normalization condition, Eq. (\ref{eq:norm}), in the last step.
Condition (\ref{eq:norm}) implies also that this quantity is bounded away from zero uniformly in $R$
for $R$ larger than some $R_0>0.$ Since the term in the second line of (\ref{eq:zexptot}) is made 
arbitrarily small by taking $R$ large enough, the proof is complete.


\section*{Appendix B. Integration of the hard rod hierarchy} \label{app:rod}

In this appendix we shall find the unique and {\em explicit} solution to the one--dimensional hard core hierarchy
({\em hard rod BBGKY hierarchy})
\bea
&& \frac{\partial\r_n}{\partial q_1}(q_1,\cdots,q_n) = \c(|q_1-a-q_i|\geq d) \r_{n+1}(q_1,\cdots,q_n,q_1-d)\nn\\
&&\ \ \ \ \ \ \ \ \ \ -\c(|q_1+a-q_i|\geq d)\r_{n+1}(q_1,q_2,\cdots,q_n,q_1+d)\;, \nn\\
&& (q_1,\cdots,q_n)\in\RRR^n_d\;,\nn\\
&& \c(\AA) = 1 \mbox{ if $\AA$ is verified and $0$ otherwise,} \label{eq:bbgky1d}
\eea
with the assumptions of invariance under translation and permutation of particles, 
sufficiently small $\r \equiv \r_1$ (precisely $\r < 1/d$), cluster property (\ref{eq:DIS}), 
continuity over $\ol\RRR^n_d,$ piecewise $C^1$ regularity on $\RRR^n_d$ and
boundedness of the derivative.
The special feature of this case is the existence of an explicit form for the equilibrium
correlation functions, e.g. \cite{LM66}. In what follows, we derive these expressions 
from the hierarchy by direct integration and without going through the corresponding
Kirkwood--Salsburg equations.

To do this, we can follow the procedure of \cite{GV75} in a rather natural 
way by ordering the particles from left to right: $q_i \leq q_{i+1}-d$; hence we start rewriting
\bea
&& \frac{\partial\r_n}{\partial q_1}(q_1,\cdots,q_n) = \r_{n+1}(q_1-d,q_1,\cdots,q_n)\nn\\
&&\ \ \ \ \ \ \ \ \ \ -\c(q_1 \leq q_2-2d)\r_{n+1}(q_1,q_1+d,q_2,\cdots,q_n)\;, \nn\\
&& q_i\in\RRR\;,\ \ \ \ \ q_i \leq q_{i+1}-d\;.
\eea
\ni Now we choose $q_0<<q_1$ and we integrate from $q_0$ to $q_1$:
\bea
&& \r_n(q_1,q_2,\cdots,q_n) = \r_n(q_0,q_2,\cdots,q_n) \nn\\
&& \ \ \ \ \ \ \ \ \ \ + \int_{q_0}^{q_1}d\bq\Big( \r_{n+1}(\bq-d,\bq,q_2\cdots,q_n)
-\c(\bq \leq q_2-2d)\r_{n+1}(\bq,\bq+d,q_2,\cdots,q_n)\Big) \nn\\
&& \ \ \ \ \ = \r_n(q_0,q_2,\cdots,q_n)+ \int_{q_0}^{q_{0}+d}d\bq \r_{n+1}(\bq-d,\bq,q_2\cdots,q_n)\nn\\
&& \ \ \ \ \ \ \ \ \ \ -\int_{q_1}^{q_{1}+d}d\bq\c(\bq \leq q_2-d)\r_{n+1}(\bq-d,\bq,q_2,\cdots,q_n)\;,
\label{eq:1dcrucial}
\eea
where we used again the symmetry in the particle labels to split the integral in the second equality.
\ni Sending $q_0$ to $- \infty$ gives
\bea
&&\r_n(q_1,q_2,\cdots,q_n) = \Big(\r+d\r_2(d)\Big)\r_{n-1}(q_2,\cdots,q_n)\nn\\
&&\ \ \ \ \ \ \ \ \ \ -\int_{q_1}^{q_{1}+d}d\bq\c(\bq \leq q_2-d)\r_{n+1}(\bq-d,\bq,q_2,\cdots,q_n)\;,\nn\\
&& \r_2(d) := \r_2(q,q+d)\;,
\eea
having used the cluster property and the translation invariance.

Call $R := \r+d\r_2(d)$. Iterating once the above equation we have
\bea
\r_n(q_1,q_2,\cdots,q_n) = R\Big[\r_{n-1}(q_2,\cdots,q_n)
-\int_{q_1}^{q_{1}+d}d\bq\c(\bq \leq q_2-d)\r_n(\bq,q_2,\cdots,q_n)\Big]\;. \label{eq:1dh}
\eea

\ni We stress again that the above explained procedure does not lead directly to the Kirkwood--Salsburg equations.
The extracted constant $R$ is different from the activity of the hard rod gas (which is known
to be given by $z= Re^{Rd}$, see for instance \cite{LM66}). 
Nevertheless the set of equations (\ref{eq:1dh}) can be solved explicitly for every $n$, 
starting from $n=2$ (the equation for $n=1$ is of course useless in this model), as we show below.
Actually, the simple structure of Eq. (\ref{eq:1dh}) allows to construct easily $\r_n$ from $\r_{n-1}:$
this structure is due to the strong symmetry used to split the integral in the second equality of (\ref{eq:1dcrucial}),
and it seems to have no analogue in higher dimensions.

We start with the $n=2$ case. Call $x = |q_2-q_1|$, $x\geq d$. Formula (\ref{eq:1dh}) implies
$\r_2(d) = \frac{\r^2}{1-\r d},$ $R=\frac{\r}{1-\r d}$ and
\bea
&&\frac{d\r_2}{dx}(x) = -R\r_2(x)\ \ \ \ \ \ \ \ \ \ d<x<2d\nn\\
&&\frac{d\r_2}{dx}(x) = R\Big(-\r_2(x)+\r_2(x-a)\Big)\ \ \ \ \ \ \ \ \ \ 2d<x\;.
\eea
Solving these set of equations iteratively in the intervals $(kd,(k+1)d), k=1,2,\cdots,$ using the continuity
assumption, leads to
\bea
\r_2(x) = \r\sum_{k=1}^{[x/d]}\Big(\frac{\r}{1-\r d}\Big)^k \frac{(x-kd)^{k-1}}{(k-1)!}
e^{-\frac{(x-kd)\r}{1-\r d}}\;. \label{eq:r2x}
\eea

In a similar way, using \eqref{eq:1dh} and \eqref{eq:r2x} and proceeding by induction on $n$, one finds
that the solution of (\ref{eq:bbgky1d}) for $n\geq 2$ is 
\bea
\r_n(q_1,\cdots,q_n) = \frac{1}{\r^{n-2}}\prod_{j=1}^{n-1}\r_2(q_j,q_{j+1})\;. \label{eq:rn1d}
\eea
%


\section*{Appendix C. The method of \cite{GV75}} \label{app:GV}

In this appendix we discuss the method established in \cite{GV75} for the integration of the hierarchy
(\ref{eq:grad-rho}), pointing out an error in the formula for the activity and sketching how to correct it
(we refer to \cite{Si11} for details). 
At the end of the section we make comparisons between this method and the one established in the present paper.

We will need somewhat stronger assumptions than those of Theorem \ref{thm:main},
namely the potential is a function $\f\in C^1(\RRR^\n)$ which is radial, stable and with compact support,
while the smooth Maxwellian state has positional correlation functions $\r_n\in C^1(\RRR^{\n n})$\ \footnote{The 
assumptions on the smoothness of $\f$ could be released as done in Section \ref{sec:setup}, by using Equations
(\ref{eq:Maxwell})--(\ref{eq:gradrnbound}).}with the following properties:
\bea
&& \mbox{{\em a)\ \ }}\r_n \leq \xi^n, |\nabla_{q_i}\r_n|\leq C_n\xi^n, \mbox{ with $\xi$ small enough;}\nn\\
&& \mbox{{\em b)\ \ }}\r_n \mbox{ is translation and rotation invariant;}\nn\\
&& \mbox{{\em c)\ \ }}\r_n \mbox{ satisfies an exponential strong cluster property, i.e.}\nn\\
&&\ \ \ \ \ \ \ \ \ \ |\r_n^{T}(q_1,\cdots,q_n)| \leq (C\xi)^n e^{-\kappa|q_1-q_n|}\;,\ \ \ \ \ C, \kappa > 0\;,
\label{eq:GVass}
\eea
where the {\em truncated correlation functions} $\r_n^T$ can be defined by
\[
\left\{
\begin{array}{l}
\r_2^{T}(q_1,q_2) = \r_2(q_1,q_2) - \r(q_1)\r(q_2)\;,  \\
\r_3^{T}(q_1,q_2,q_3) = \r_3(q_1,q_2,q_3) - \r_2(q_1,q_2)\r(q_3)-\r(q_1)\r_2(q_2,q_3)+\r(q_1)\r(q_2)\r(q_3)\;,\\
\end{array}
\right.
\]
etc. (see \cite{DIS74} and \cite{GV75}, page 279).

Let us start by integrating Eq. (\ref{eq:grad-rho}) along a straight line connecting $q_0$ (arbitrary) to $q_1:$ using the 
same notations introduced for (\ref{eq:inthat}), we have
\bea
&&\r_n(q_1,\cdots,q_n) = e^{-\b\left(W_{q_1}(q_2,\cdots,q_n)-W_{q_0}(q_2,\cdots,q_n)\right)}
\r_{n}(q_0,q_2,\cdots,q_n) \label{eq:GVbah}\\
&&\ \ \ \ \ +\int_{q_0}^{q_1} d\bq_1\int_{\RRR^\n} dy_1 \frac{\partial(-\b\f(\bq_1-y_1))}{\partial\bq_1}
e^{-\b\left(W_{q_1}(q_2,\cdots,q_n)-W_{\bq_1}(q_2,\cdots,q_n)\right)}\r_{n+1}(\bq_1,q_2,\cdots,q_n,y_1)\;,\nn
\eea
which is nothing but a rewriting of Eq. (\ref{eq:inthat}). In the assumption {\em b)}, the case $n=1$ is a trivial identity, hence
we shall assume $n\geq 2$ in the following.

The strategy consists in {\em taking the limit as $|q_0|\rightarrow +\infty$ right away, before iteration of formulas}. 
This is also the essential difference with respect to the method discussed in Section \ref{sec:mr}, where such a limit is taken
at the very end of the proof, after infinitely many iterations. Using {\em c)}, from (\ref{eq:GVbah}) we get
\bea
&&\r_n(q_1,\cdots,q_n) = \r e^{-\b W_{q_1}(q_2,\cdots,q_n)}\r_{n-1}(q_2,\cdots,q_n)\label{eq:GV7ex}\\
&&\ \ \ \ \ + e^{-\b W_{q_1}(q_2,\cdots,q_n)}
\int_{-\infty}^{q_1} d\bq_1\int_{\RRR^\n} dy_1 \frac{\partial(-\b\f(\bq_1-y_1))}{\partial\bq_1}
e^{\b W_{\bq_1}(q_2,\cdots,q_n)}\r_{n+1}(\bq_1,q_2,\cdots,q_n,y_1)\;,\nn
\eea
where the double integral in the second term on the right hand side is well defined
(by the exponential clustering (\ref{eq:GVass}) and the assumed rotation symmetry of the potential and of $\r_2$),
though not absolutely convergent. Interchanging the integrations we find
\bea
&&\r_n(q_1,\cdots,q_n) = (\r-\g)e^{-\b W_{q_1}(q_2,\cdots,q_n)}\r_{n-1}(q_2,\cdots,q_n)\label{eq:GV7}\\
&&\ \ \ \ \ + e^{-\b W_{q_1}(q_2,\cdots,q_n)} \int_{\RRR^\n} dy_1
\int_{-\infty}^{q_1} d\bq_1 \frac{\partial(-\b\f(\bq_1-y_1))}{\partial\bq_1}
e^{\b W_{\bq_1}(q_2,\cdots,q_n)}\r_{n+1}(\bq_1,q_2,\cdots,q_n,y_1)\;,\nn 
\eea
where we put
\bea
&&\g =  \int_{\RRR^\n} dy_1 \int_{-\infty}^{q_1} d\bq_1
\r_2(\bq_1,y_1) \frac{\partial(-\b\f(\bq_1-y_1))}{\partial\bq_1}\nn\\
&& \equiv \int_{B(q_1)} dy_1 \int_{-\infty}^{q_1} d\bq_1
\r_2(\bq_1,y_1) \frac{\partial(-\b\f(\bq_1-y_1))}{\partial\bq_1}\;. 
\eea
Here $B(q)$ is the ball centered in $q$ and with radius equal to the range of $\f,$ and the second equality 
is again true by the rotation symmetry. 

The authors in \cite{GV75} proceed by iteration of formula (\ref{eq:GV7}). Call for simplicity 
\be
\z = \r -\g\;.
\ee
The first iteration gives
\bea
&&\r_n(q_1,\cdots,q_n) = \z e^{-\b W_{q_1}(q_2,\cdots,q_n)}\Big[\r_{n-1}(q_2,\cdots,q_n)\label{eq:GV81} \\
&&\ \ \ \ \ \ \ \ \ \ - \int_{\RRR^\n} dy_1(1-e^{-\b\f(q_1-y_1)})\r_{n}(q_2,\cdots,q_n,y_1)\Big]\nn\\
&&\ \ \ \ \ +e^{-\b W_{q_1}(q_2,\dots,q_n)}\int_{\RRR^\n} dy_1\int_{-\infty}^{q_1} d \bar{q}_1\int_{\RRR^\n} dy_2 \int_{-\infty}^{\bar{q}_1} d\bar{q}_2  \prod_{j=1}^{2}\left( \frac{\partial}{\partial\bar{q}_j}(1-e^{-\b\f(\bar{q}_j-y_j)})\right) \nn\\
&&\ \ \ \ \ \ \ \ \ \ \cdot e^{\b W_{\bar{q}_{2}}(q_2,\dots,q_n,y_1,y_2)}\r_{n+2}(\bar{q}_{2},q_2,\dots,q_n,y_1,y_2)\;.\nn
\eea
If we iterate once again (\ref{eq:GV7}) in (\ref{eq:GV81}), the last term of (\ref{eq:GV81})
becomes
\bea
&& \z e^{-\b W_{q_1}(q_2,\dots,q_n)}\int_{\RRR^\n} dy_1\int_{-\infty}^{q_1} d \bar{q}_1\int_{\RRR^\n} 
dy_2 \int_{-\infty}^{\bar{q}_1} d\bar{q}_2 \prod_{j=1}^{2}
\left( \frac{\partial}{\partial\bar{q}_j}(1-e^{-\b\f(\bar{q}_j-y_j)})\right) \label{eq:GV82}\\
&&\ \ \ \ \ \ \ \ \ \  \cdot \r_{n+1}(q_2,\dots,q_n,y_1,y_2) \nn\\
&& \ \ \ \ \ -e^{-\b W_{q_1}(q_2,\dots,q_n)}\int_{\RRR^\n} dy_1\int_{-\infty}^{q_1} d \bar{q}_1\int_{\RRR^\n} 
dy_2 \int_{-\infty}^{\bar{q}_1} d\bar{q}_2\int_{\RRR^\n} dy_3 \int_{-\infty}^{\bar{q}_2} d\bar{q}_3  \nn\\
&&\ \ \ \ \ \ \ \ \ \ \cdot\prod_{j=1}^{3}\left( \frac{\partial}{\partial\bar{q}_j}(1-e^{-\b\f(\bar{q}_j-y_j)})\right)
e^{\b W_{\bar{q}_{3}}(q_2,\dots,q_n,y_1,y_2,y_3)}\r_{n+3}(\bar{q}_{3},q_2,\dots,q_n,y_1,y_2,y_3)\nn\;,
\eea
which is not equal to the formula in step 8) of \cite{GV75} with $N=2$, since the first term of (\ref{eq:GV82})
is not equal to $\z e^{-\b W_{q_1}(q_2,\dots,q_n)}$ times
\bea
&&\frac{1}{2}\int_{\RRR^\n} dy_1\int_{\RRR^\n} dy_2 
(1-e^{-\b\f(q_1-y_1)})(1-e^{-\b\f(q_1-y_2)}) \r_{n+1}(q_2,\dots,q_n,y_1,y_2)\;.\label{eq:KS2}
\eea
In fact, integrating in $\bar q_2$ the first term of (\ref{eq:GV82}), we have $\z e^{-\b W_{q_1}(q_2,\dots,q_n)}$ times
\be
\int_{\RRR^\n} dy_1  \int_{\RRR^\n} dy_2 \int_{-\infty}^{q_1} d \bar q_1 (1-e^{-\b\f(\bar q_1-y_2)})
\frac{\partial(1-e^{-\b\f(\bar q_1 -y_1)})}{\partial\bar q_1}\r_{n+1}(q_2,\dots,q_n,y_1,y_2)\label{eq:m}\\
\ee
and, since the integrals in $y_1$ and $y_2$ are not interchangeable, integration by parts of this formula
does not lead to the single desired term (\ref{eq:KS2}), in spite of the symmetry of $\r_{n+1}$. We refer to 
\cite{Si11} for a proof of the last assertion.

The additional terms missing in the formula in step 8) of \cite{GV75}, obtained by repeated iteration of (\ref{eq:GV7}),
give higher order corrections to the constant $\z$, and all these (infinitely many) corrections lead to a definition of 
the activity. 
A rather convenient way to modify the proof (leaving essentially unchanged the procedure)
in order to obtain the correct expression for the activity, is to keep the inverse order of integration in formula 
(\ref{eq:GV7}): that is to iterate Eq. (\ref{eq:GV7ex}). This computation is reported in detail in \cite{Si11}.
The convergence of the iteration procedure is handled in assumptions 
{\em a)}, {\em b)} and {\em c)}: after $N$ iterations one gets a remainder $R_{n,N}$ that can be bounded as
\be
|R_{n,N}(q_1,\dots,q_n)| \leq (A\xi)^{n+N+1}\;, \label{eq:GVremterr}
\ee
where $A$ is a suitable constant depending on $\b, C, \kappa,\f$ and the configuration $q_1,\cdots,q_n$ (but not on $\xi$), 
so that it goes to zero when $N\rightarrow\infty$ if $\xi$ is small enough. For these values of $\xi,$
the method provides convergence to the Kirkwood--Salsburg equations with exponential rate.

{\bf Comparisons.} We have the following differences with respect to the method presented in Section~\ref{sec:mr}:
\begin{enumerate}
\item The rate of convergence of the iteration is exponential, Eq. (\ref{eq:GVremterr}) (instead of factorial,
see (Eq. (\ref{eq:rem})): this implies convergence for sufficiently small values of $\xi.$ 
\item The radius of convergence of the procedure is at least $1/A,$ where $A$ is not uniformly bounded in the maximum 
of the potential; a bound which is uniform in the hard core limit could be obtained by assuming estimates for
the smooth state as those in (\ref{eq:rnbound})--(\ref{eq:gradrnbound}), instead of {\em a)} of (\ref{eq:GVass}).
\item The exponential strong cluster property (\ref{eq:GVass}) (instead of the weak cluster (\ref{eq:DIS})), 
as well as the short range assumption on the potential (instead of the weak decrease (\ref{eq:decrease})), 
are needed to control the convergence of the integrals over the unbounded domains of integration;
\item the same can be said for the assumption {\em b)} in (\ref{eq:GVass}), that is for the rotation invariance of the state, 
which is not needed in Theorem \ref{thm:main}. 
Furthermore we shall notice that, in the proof of Theorem \ref{thm:main}, translation invariance is only used in the very 
last step, i.e. to perform the limit $|q_0|\rightarrow +\infty$ of expression (\ref{eq:KSq0}), obtained after infinitely 
many iterations. This makes the method suitable to extend the result to situations in which there
is no symmetry and different kinds of boundary condition are considered; an example of such a situation has been
given in Theorem~\ref{thm:infcont}.
\end{enumerate}

\end{document}